\documentclass[letterpaper, 10 pt, conference, twoside]{ieeeconf}

\IEEEoverridecommandlockouts                              

\usepackage[pdftex, pdfstartview={FitV}, pdfpagelayout={TwoColumnLeft},bookmarksopen=true,plainpages = false, colorlinks=true, linkcolor=black, citecolor = black, urlcolor = black,filecolor=black , pagebackref=false,hypertexnames=false, plainpages=false, pdfpagelabels ]{hyperref}

\usepackage{setspace}
\usepackage{graphicx}
\usepackage{epstopdf}
\usepackage{amssymb,amsmath}

\usepackage{amsthm}
\usepackage{mathtools}
\usepackage{cuted}
\usepackage{cases}
\usepackage[font=footnotesize]{subcaption}

\usepackage[font=footnotesize]{caption}
\usepackage{xcolor}
\usepackage{units}
\usepackage{algorithm}
\usepackage[noend]{algpseudocode}
\usepackage{balance}
\usepackage[sort,compress,noadjust]{cite}
\usepackage{tabularx}
\usepackage{nameref}
\usepackage{centernot}

\usepackage{arydshln}
\newtheoremstyle{definition}{}{}{}{}{\bfseries}{.}{.5em}{\thmname{#1}\thmnumber{ #2}\thmnote{ (#3)}}
\theoremstyle{definition}
\newtheorem{definition}{Definition}

\let\labelindent\relax
\usepackage{enumitem}

\theoremstyle{plain}
\newtheorem{theorem}{Theorem}
\theoremstyle{plain}
\newtheorem{proposition}{Proposition}
\theoremstyle{plain}
\newtheorem{lemma}{Lemma}
\theoremstyle{plain}
\newtheorem{corollary}{Corollary}
\theoremstyle{definition}

\theoremstyle{remark}
\newtheorem{remark}{Remark}

\usepackage[hang,flushmargin]{footmisc}

\usepackage[capitalize]{cleveref}
\crefformat{equation}{(#2#1#3)}
\Crefformat{equation}{Equation~(#2#1#3)}
\Crefname{equation}{Equation}{Eqs.}

\usepackage{accents}

\title{\LARGE \bf New Insights into Cascaded Geometric Flight Control: \\ From Performance Guarantees to Practical Pitfalls}

\author{Brett T. Lopez$^1$
\thanks{$^{1}$VECTR Laboratory, University of California, Los Angeles, Los Angeles, CA, {\tt\small btlopez@ucla.edu}}}

\begin{document}

\maketitle

\begin{abstract}
   We present a new stability proof for cascaded geometric control used by aerial vehicles tracking time-varying position trajectories. Our approach uses sliding variables and a recently proposed quaternion-based sliding controller to demonstrate that exponentially convergent position trajectory tracking is theoretically possible. Notably, our analysis reveals new aspects of the control strategy, including how tracking error in the attitude loop influences the position loop, how model uncertainties affect the closed-loop system, and the practical pitfalls of the control architecture.
\end{abstract}

\section{Introduction}

Autonomous aerospace systems rely on precise trajectory tracking to safely execute maneuvers that satisfy state and actuator constraints.
However, designing high-performance tracking controllers is nontrivial due to the inherent nonlinearities arising from orientation dynamics and the strong coupling between rotational and translational dynamics.
The seminal work by Frazzoli et al.~\cite{frazzoli2000trajectory} was the first to  employ ``backstepping on manifolds'' design to prove closed-loop asymptotic stability of a cascaded geometric controller (see \cref{fig:geo_cascaded_control}) for VTOL aircraft tracking time-varying position trajectories.
By representing vehicle orientation as a rotation matrix belonging to the special orthogonal group $SO(3)$, the approach in \cite{frazzoli2000trajectory} combined control design on smooth manifolds with nested feedback loops, i.e., backstepping, to get, using modern terminology, an asymptotically convergent closed-loop system \cite{pavlov2004convergent,pavlov2005convergent}.
In other words, the closed-loop system was guaranteed to track any dynamically feasible time-varying position trajectory asymptotically. 
More broadly, \cite{frazzoli2000trajectory} was the first to propose a nested tracking control law designed on the special Euclidean manifold $SE(3) = \mathbb{R}^3 \times SO(3)$ for VTOL aircraft.

The popularity of small-scale VTOL multirotors in the early 2010's renewed interest in nonlinear flight control and trajectory tracking.
Lee et al.~\cite{lee2010geometric,lee2012exponential,goodarzi2013geometric} extended \cite{frazzoli2000trajectory} by combining feedback linearization, geometric control, and backstepping to attain almost global exponential stability on $SE(3)$ with a similar architecture as \cite{frazzoli2000trajectory}, but featuring a more effective attitude tracking controller.
Almost global stability is the best possible result attainable when using a rotation matrix attitude parameterization \cite{bhat2000topological}.
The attitude controller from Lee et al.~was strengthened in \cite{lopez2021sliding} using a quaternion sliding variable to achieve globally exponentially convergent trajectory tracking, but the approach was not examined within the context of cascaded geometric control.

The main contributions of this work are a new stability proof for aerial vehicles tracking time-varying trajectories by combining cascaded geometric control with sliding variables \cite{slotine1991applied}, along with the new insights gained from using sliding variables in the analysis.
Building on the recently proposed quaternion-based sliding variable \cite{lopez2021sliding} that guarantees global exponential convergence to any desired orientation trajectory, we prove that nesting the attitude controller from \cite{lopez2021sliding} with linear position control yields an exponentially convergent closed-loop system \cite{pavlov2004convergent,pavlov2005convergent}, which formally means that the position tracking error converges to zero exponentially.
By employing sliding variables in our analysis, the stability proof is streamlined, and, more importantly, we obtain a differential inequality that elucidates the complex coupling between attitude and position tracking. 
Specifically, our analysis reveals how the tracking error in the attitude control loop influences the position control loop; information that can be used, e.g., in a robust MPC framework for provably-safe planning with model uncertainty \cite{langson2004robust,mayne2011tube,lopez2019dynamic}.
Other theoretical aspects of cascaded geometric control are also revealed, such as the robustness to parametric and non-parametric model uncertainties and how the control architecture eliminates unmatched uncertainties. 
We also discuss the practical limitations of the control architecture, including the need for acceleration and jerk feedback and the differentiation of feedforward models, all of which may suggest the need for new tracking controllers for highly dynamic aerial vehicles.

\textit{Notation:} 
For two vectors $a,b \in \mathbb{R}^3$, the cross product matrix is $[\cdot]_\times$ where $a \times b = [a]_\times b$.
For any vector $w \in \mathbb{R}^n$, the Euclidean norm is $\|w \|^2 = w^\top w$.
The Frobenius norm is an extension of the Euclidean norm to matrices where for $A \in \mathbb{R}^{m\times n}$ then $\|A\|^2_F = {\mathrm{tr}\{A A^\top\}} = {\sum_i^m \sum_j^n a_{ij}^2} = {\sum_i^m \|a_i\|^2}$ with $a_{ij}$ being the $ij$ element of $A$ and $a_i^\top \in \mathbb{R}^{n}$ is the $i$-th row vector of $A$. 
The normalized unit vector of $u \in \mathbb{R}^n$ with $\|u\| \neq 0$ is $\hat{u} = u / \|u\|$.
The upper right Dini derivative of a non-differentiable real-valued function $v$ is $D^+v(t) = \limsup_{h\rightarrow0^+} (v(t+h)-v(t))/h$.

\section{Preliminaries}
The rigid body dynamics of the class of aerial vehicles under consideration are
\begin{equation}
    \label{eq:dyn}
    \begin{aligned}
        m \ddot{r} & = R \, T \, \hat{e}_3 - m g \\
        \dot{R} & = R \, [\omega]_\times \\
        J \dot{\omega} & = - \omega \times J \omega + \tau 
    \end{aligned}
\end{equation}
where $r \in \mathbb{R}^3$ is the position vector, $R \in SO(3) = \{ R \in \mathbb{R}^{3\times 3} \, | \, R R^\top = R^\top R = I, ~ \det(R) = 1\}$ is the rotation matrix belonging to the special orthogonal group that describes the orientation of the vehicle's body frame with respect to an inertial frame, $\omega \in \mathbb{R}^3$ is the body angular velocity vector, and $\hat{e}_3 = [0,\,0,\,1]^\top$.  
The thrust vector is $T \, \hat{e}_3 \in \mathbb{R}^3$ (assumed to be along the body-frame z-axis) and the body torque vector is $\tau \in \mathbb{R}^3$; both are considered control inputs to the system \cref{eq:dyn}.
The model parameters are the vehicle's mass $m$, its moment of inertia tensor $J$, and the gravity vector $g \in \mathbb{R}^3$.

An equivalent (albeit non-unique) representation of the vehicle's orientation is the unit quaternion $q \in \mathbb{S}^3$ that lives on the 3-sphere. 
Using an ambient Euclidean space embedding, we have $\mathbb{S}^3 = \{ q \in \mathbb{R}^4\, | \, q^\top q = 1 \}$.
The unit quaternion can be interpreted as a 4-dimensional complex number with a real part $q^\circ$ and vector (imaginary) part $\vec{q}$ where $q = (q^\circ, \vec{q}\,) = q^\circ + q_x {i} + q_y {j} + q_z {k}$ with complex numbers ${i},{j},{k}$ that satisfy $i^2 = j^2 = k^2 = i j k = -1$.
Another way to express a unit quaternion is through the axis-angle representation.
If $\theta \in \mathbb{R}$ is the rotation angle and $\hat{n} \in \mathbb{R}^3$ is the rotation axis, then $q = ( \cos(\theta/2), \hat{n} \sin(\theta/2))$.

The set of unit quaternions forms a Lie group with the quaternion multiplication operator $\otimes$.
Given two unit quaternions $q, p \in \mathbb{S}^3$, we have $q \otimes p = (q^\circ p^\circ - \vec{q}^\top \vec{p}, \, q^\circ \vec{p} + p^\circ \vec{q} + \vec{q} \times \vec{p}\,).$ 
The conjugate quaternion is $q^* = (q^\circ, -\vec{q}\, )$, and it can be shown to be the inverse of quaternion $q$ since $q \otimes q^* = q^* \otimes q = (1, \vec{0}\,)$.
The kinematics of $q$ are $\dot{q} = \tfrac{1}{2} \, q \otimes (0, \omega)$. 
A unit quaternion $q$ is related to a rotation matrix $R$ via $R = I + 2 q^\circ [\,\vec{q}\,]_\times + 2 [\,\vec{q}\,]_\times^2 = (1-2 \|\,\vec{q}\,\|^2) I + 2 q^\circ [\,\vec{q}\,]_\times + 2 \vec{q} \, \vec{q}^\top$.
This relation can be derived via Rodrigues' rotation formula with the axis-angle representation, namely $R = I + \sin\theta [\hat{n}]_\times + (1-\cos\theta)[\hat{n}]^2_\times$.
Note that $q$ and $-q$ give the same rotation matrix $R$.
This is known as the double cover property and must be taken into account when using quaternions for control or estimation to avoid the unwinding phenomenon \cite{mayhew2011quaternion}.

\begin{figure}[t!]
    \vspace{0.1in}
    \centering
    \includegraphics[width=1.0\columnwidth]{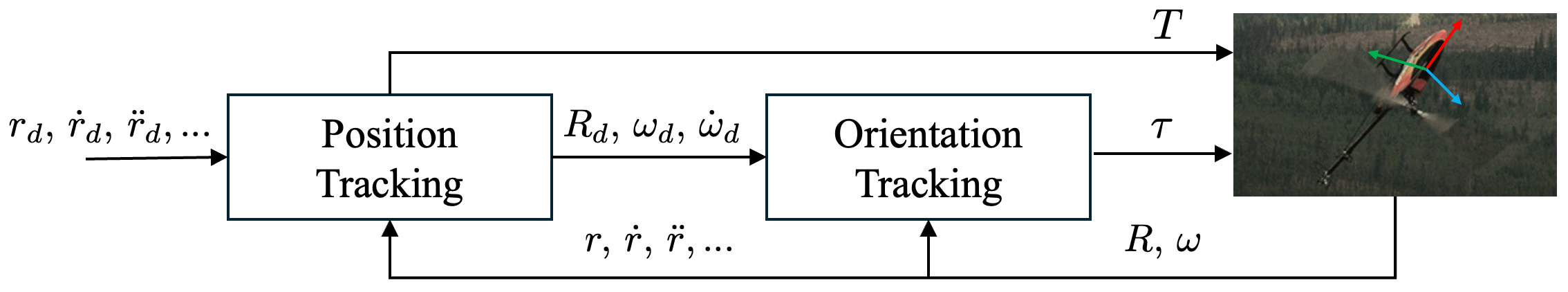}
    \caption{Cascaded geometric control architecture for aerial vehicles.}
    \label{fig:geo_cascaded_control}
    \vspace{-0.2in}
\end{figure}

\section{Main Results}
\label{sec:results}

\subsection{Overview}
This section contains a new analysis of cascaded geometric control for tracking a time-varying desired trajectory with sliding variables.
First, a globally exponentially convergent attitude controller is presented that uses a quaternion-based sliding variable proposed in \cite{lopez2021sliding}.
Second, an almost globally exponentially convergent position controller is presented that utilizes sliding variables to reveal new insights into the coupling between the inner- and outer-loop controllers.
Finally, we discuss the controller's theoretical guarantees and highlight some of its practical limitations/pitfalls. 
Before proceeding to the main technical results, we define the notion of a \emph{convergent system} \cite{pavlov2004convergent,pavlov2005convergent}, which is a useful form of stability when concerned with trajectory tracking.

\begin{definition}[\cite{pavlov2004convergent}]
\label{def:convergent}
    A dynamical systems is asymptotically (resp. exponentially) convergent if there exists a bounded time-varying solution such that the system asymptomatically (resp. exponentially) converges to.
\end{definition}

\begin{remark}
    The idea of a convergent system is closely related to incremental stability \cite{angeli2002lyapunov} and contraction theory \cite{lohmiller1998contraction}, the concepts being equivalent when the system evolves on a compact manifold. 
    However, because this work focuses on trajectory tracking rather than stability of neighboring trajectories, results are expressed using terminology from \cite{pavlov2004convergent}.
\end{remark}

\subsection{Globally Exponentially Convergent Attitude Control}
The sliding variable defined in \cite{lopez2021sliding} can be used to obtain exponentially convergent attitude dynamics. 
First, let the error quaternion be $q_e = q_d^* \otimes q$ where $q_d^*$ is the conjugate of the desired orientation and $q$ is the current orientation.
The error quaternion kinematics can then be shown to be
\begin{equation}
    \label{eq:qe_dynamics}
    \dot{q}_e = \frac{1}{2} q_e \otimes (\, 0, \, {\omega_e} \, ) =  \frac{1}{2} ( \, -\vec{q}_e^\top {\omega_e}, ~ q_e^\circ {\omega_e} + \vec{q}_e \times {\omega_e} \, ),
\end{equation}
where the angular velocity error is ${\omega_e} = {\omega} -  R_e^\top {\omega_d}$ with $R_e$ being the error rotation matrix formed from $q_e$. 
Note that $R_e$ is needed because ${\omega_d}$ and ${\omega}$ are tangent vectors on the three-sphere at \emph{different} locations. 
As a result, ${\omega_d}$ must be transformed into the same frame as ${\omega}$ for the vector difference to be meaningful.

Let the quaternion sliding variable ${s} \in \mathbb{R}^3$ be defined as
\begin{equation}
    \label{eq:s_quat}
    \begin{aligned}
        s & = {\omega_e} + 2\, \lambda  \, \mathrm{sgn}\left(q_e^\circ\right)\vec{q}_e \\ 
        & = \omega - R_e^\top \omega_d + 2\, \lambda  \, \mathrm{sgn}\left(q_e^\circ\right)\vec{q}_e \\
        & = \omega - \omega_r
    \end{aligned}
\end{equation}
where $\omega_r$ is obviously defined, $\lambda \in \mathbb{R}_{>0}$, and $\mathrm{sgn}(x) = x / |x|$ for $x \neq 0$ and $\mathrm{sgn}(0) = 1$.
Observe $\mathrm{sgn}(\cdot)$ differs from its standard definition; a distinction important for establishing global exponential convergence.
\cref{proposition:qe_local} shows that if ${s} = 0$ indefinitely via feedback then $\|\vec{q}_e\| \rightarrow 0$ exponentially at rate $\lambda$, which is equivalent to $q \rightarrow \pm q_d$ exponentially.

\begin{proposition}
    \label{proposition:qe_local}
    Let the manifold ${\mathcal{S}} = \{({q}_e,\omega_e) \in \mathbb{S}^3 \times \mathbb{R}^3 \, | \, {s({q}_e,\omega_e)}=0\}$ be invariant. Then, for any trajectory initialized on ${\mathcal{S}}$, the vector part of $q_e$ globally exponentially converges to zero with rate $\lambda$.
\end{proposition}

\begin{proof}
    Consider the Lyapunov function $V = \|\vec{q}_e\|^2$.
    Differentiating along the quaternion error dynamics \cref{eq:qe_dynamics} yields,
    \begin{equation}
        \dot{V} = 2 \vec{q}_e^\top\dot{\vec{q}}_e = \vec{q}_e^\top \left( q_e^\circ \omega_{e} + \omega_{e} \times \vec{q}_e \right) = q_e^\circ \vec{q}_e^\top \omega_{e},
        \label{eq:v_dot_local}
    \end{equation}
    where the third equality is obtained by noting that $\vec{q}_e^\top \omega_e \times \vec{q}_e = 0$.
    Since the manifold ${\mathcal{S}} $ is invariant, ${\omega_e} = -2 \, \lambda  \, \mathrm{sgn}\left(q_e^\circ\right) \, \vec{q}_e$ from \cref{eq:s_quat}, so \cref{eq:v_dot_local} becomes
    \begin{equation}
        \label{eq:v_dot_qe_local}
        \dot{V} = - 2 \, \lambda \, |q_e^\circ| \,  \|\vec{q}_e\|^2 = - 2 \, \lambda \,  V \, \sqrt{1-V},
    \end{equation}
    where we used the fact that $q_e$ is a unit quaternion. 
    To show exponential convergence, we can use \cref{lemma:ode} by replacing $x$ with $V$ and $\sigma$ with $2 \, \lambda$.
    In doing so, we get $V(t) \leq 4 e^{-2\lambda t} \implies \|\vec{q}_e(t)\| \leq 2 e^{-\lambda t}$ for $\|\vec{q}_e \| \neq 1$, which is a cleaner but looser bound than stated in \cref{lemma:ode}.
    Hence $\|\vec{q}_e(t)\| \rightarrow 0$ exponentially at rate $\lambda$.
    To establish global convergence, the radially unboundedness criterion cannot be used since $q_e$ lives on a compact manifold.
    It is sufficient to show that $\|\vec{q}_e\| = 1$, i.e., when $q_e^\circ = 0$, is not an equilibrium point. 
    The closed-loop kinematics for $q_e^\circ$ are $\dot{q}_e^\circ =  \lambda \, \mathrm{sgn}(q_e^\circ) \, \| \vec{q}_e \|^2$, which is non-zero when $q_e^\circ = 0$ (i.e., $\|\vec{q}_e\| = 1$) by the modified definition of $\mathrm{sgn}(\cdot)$. 
    Therefore, for any trajectory initialized on the invariant manifold $\mathcal{S}$ then $\|\vec{q}_e\|$ globally exponentially converges to zero as desired.
\end{proof}

The premise of \cref{proposition:qe_local} is that the manifold $\mathcal{S} = \{({q}_e,\omega_e) \in \mathbb{S}^3 \times \mathbb{R}^3 \, | \, {s({q}_e,\omega_e)}=0\}$ is invariant. 
This can be achieved by selecting the body torque vector $\tau$ that drive $s$ to the origin exponentially, as shown next. 

\begin{theorem}
    \label{thm:att_control}
    Let $q_d(t)$ be a smooth desired quaternion trajectory.
    The quaternion sliding variable from \cref{eq:s_quat} exponentially converges to zero with the torque control law
    \begin{equation*}
        \tau = \omega \times J \omega - J \dot{\omega}_r - K s
    \end{equation*}
    where $K \in \mathcal{S}^3_+$ is a symmetric positive definite matrix.
\end{theorem}

\begin{proof}
    Consider the candidate Lyapunov function $V = \|s\|^2 $ where $s$ is the \cref{eq:s_quat}.
    Differentiating and using \cref{eq:s_quat},
    \begin{equation*}
        \begin{aligned}
            \dot{V} &= 2 s^\top \left( \dot{\omega} - \dot{\omega}_r \right) \\
            & = 2 s^\top \left( J^{-1} (-\omega \times J \omega + \tau ) - \dot{\omega}_r  \right).
        \end{aligned}
    \end{equation*}
    Picking $\tau = \omega \times J \omega + J \dot{\omega}_r - K s$ where $K \in \mathcal{S}^3_+$, then $\dot{V} = - 2 s^\top K s$.
    Assuming there exists a scalar $\upsilon \in \mathbb{R}_{>0}$ such that $K \succeq \upsilon I$, then $\dot{V} \leq - 2 \upsilon V \implies V(t) \leq V(0) e^{-2 \upsilon t}$ so $\lim_{t\rightarrow \infty} V(t) \rightarrow 0$ exponentially at rate $2 \upsilon$. 
    Since $V = \|s\|^2$, then $\lim_{t\rightarrow \infty}\|s(t)\| \rightarrow 0$ exponentially at rate $\upsilon$.
\end{proof}

\begin{remark}
    The expanded form of the attitude control law used in \cref{thm:att_control} is
    \begin{equation}
        \label{eq:att_control}
        \tau = \omega \times J \omega + J \left(\dot{R}_e^\top \omega_d + R_e^\top \dot{\omega}_d - 2 \lambda \mathrm{sgn}(q_e^\circ) \dot{\vec{q}}_e\right) - K s,
    \end{equation}
    which \emph{guarantees} the vehicle's orientation will exponentially converge to any time-varying, dynamically feasible trajectory given by the tuple $(q_d(t), \omega_d(t))$.
    Formally, this means that the closed-loop attitude dynamics are \emph{exponentially convergent} per \cref{def:convergent}, a property that will prove to be critical for proving closed-loop position tracking via cascaded control.
    An important byproduct of this result is that the closed-loop attitude dynamics immediately inherent robustness to parametric and non-parametric uncertainties. 
    These points are summarized in the following corollary. 
\end{remark}

\begin{corollary}
    The closed-loop attitude dynamics are exponentially convergent with the torque control law \cref{eq:att_control}.
\end{corollary}
\begin{proof}
    Follows from the definition of an exponentially convergent system (see \cref{def:convergent}). 
\end{proof}

\subsection{Exponentially Convergent Cascaded Geometric Control}
With the results from the previous section in-hand, namely, global exponentially convergent attitude tracking with controller \cref{eq:att_control}, we can prove almost global exponentially convergent position tracking is achievable via cascaded control.
The analysis we present in this subsection will primarily utilize rotation matrices to keep notation manageable.
Although rotation matrices are used here, the results of the previous subsection are still pertinent.
We will adopt the same convention as before, i.e., $q \leftrightarrow R$ and $q_d \leftrightarrow R_d$, and we define the rotation matrix orientation error as $R_e = R_d^\top R$. 
The following lemma bridges the quaternion-based analysis of the previous section with rotation matrices and, in doing so, will be useful for the subsequent analysis. 

\begin{lemma}
    \label{lemma:rot_frob}
    Let $R \in SO(3)$ be a rotation matrix formed from $q = (q^\circ, \vec{q}\,)$. 
    The Frobenius norm of $R - I$ is
    \begin{equation}
        \label{eq:frob_rot}
        \|R - I\|_F = 2 \sqrt{2} \|\, \vec{q} \,\|.
    \end{equation}
\end{lemma}
\begin{proof}
    By the definition of the Frobenius norm,
    \begin{equation*}
        \begin{aligned}
            \|R - I\|^2_F & = \mathrm{tr}\{(R-I)(R-I)^\top\} \\
                    & = \mathrm{tr}\{ R R^\top - R - R^\top + I\} \\
                    & = \mathrm{tr}\{ 2 I - R - R^\top \}.
        \end{aligned}
    \end{equation*}
    Since $\mathrm{tr}\{R\} = \mathrm{tr}\{R^\top\}$ and $\mathrm{tr}\{\cdot\}$ is a linear map, then
    \begin{equation*}
        \begin{aligned}
            \|R - I\|^2_F & = 2\, \mathrm{tr}\{I\} - 2 \, \mathrm{tr}\{R\} \\
                & = 6 - 2 \, \mathrm{tr}\{ I + 2 q^\circ [ \, \vec{q} \, ]_\times + 2 [\, \vec{q}\,]_\times^2 \} \\ 
                & = 4 \, \mathrm{tr}\{ [\, \vec{q} \, ]_\times^2 \},
        \end{aligned}
    \end{equation*}
    where the second equality follows from Rodrigues' formula and the third equality follows from the fact that $\mathrm{tr}\{[\cdot]_\times\} = 0$.
    The trace of $[\, \vec{q} \, ]_\times^2$ can be calculated directly and shown to be $\mathrm{tr}\{[\,\vec{q}\,]^2_\times\} = 2\|\,\vec{q}\,\|^2$.
    Therefore, we have
    \begin{equation*}
        \|R - I\|_F^2 = 8 \|\,\vec{q}\,\|^2 \implies \|R - I\|_F = 2 \sqrt{2} \| \, \vec{q} \, \|. \qedhere
    \end{equation*}
\end{proof}

\begin{remark}
    \label{remark:rot_frob}
    An immediate extension to \cref{lemma:rot_frob} is that given another rotation matrix $\tilde{R} \in SO(3)$, then
    \begin{equation*}
        \|\tilde{R}\, R \, \tilde{R}^\top - I \|_F = 2 \sqrt{2} \, \|\,\vec{q}\,\|.
    \end{equation*}
    where again $R$ is the rotation matrix formed from $q$.
\end{remark}

\begin{remark}
    \Cref{lemma:rot_frob} essentially states that the ``distance" of a frame to the upright orientation (which is identity) is just the vector part of the quaternion that describes the frame's orientation. 
    The notion of distance is captured via the Frobenius norm, and is a deviation from the standard Euclidean distance because rotation matrices live in $SO(3)$.
\end{remark}

We will employ a position sliding variable to streamline stability analysis. 
Specifically, if we define the position tracking error to be $r_e = r - r_d$, then an obvious choice for the position sliding variable $\mathfrak{s} \in \mathbb{R}^3$ is 
\begin{equation}
    \label{eq:s_pos}
    \begin{aligned}
        \mathfrak{s} & = \dot{r}_e + \alpha \, r_e \\ 
        & = \dot{r} - \dot{r}_d + \alpha \, r_e \\ 
        & = \dot{r} - \dot{r}_r
    \end{aligned}
\end{equation}
where $\alpha \in \mathbb{R}_{>0}$.
Since \cref{eq:s_pos} is just a first-order linear system, it is trivial to prove that if the manifold $\mathrm{S} = \{ (r_e,\dot{r}_e) \in \mathbb{R}^3 \times \mathbb{R}^3 \, | \, \mathfrak{s}(r_e,\dot{r}_e) = 0 \}$ is invariant, then for any trajectory starting on $\mathrm{S}$, we have $r_e$ will converge to zero exponentially at rate $\alpha$.
More generally, as will be shown later, if $\|\mathfrak{s}\|$ converges to zero exponentially, then so too does $\|r_e\|$.

The next theorem utilizes \cref{lemma:pert_rate,lemma:pert_input} (stated in the Appendix) to prove exponentially convergent closed-loop position dynamics are achievable with \cref{eq:att_control} and a linear position controller.

\begin{theorem}
    \label{thm:pos_control}
    Let $r_d(t)$ be a smooth desired position trajectory.
    The position sliding variable from \cref{eq:s_pos} converges to zero exponentially if the closed-loop attitude dynamics are exponentially convergent, and if there exists a desired rotation matrix $R_d \in SO(3)$ and thrust $T$ such that
    \begin{equation}
        \label{eq:Rd}
        R_d \, T \, \hat{e}_3 =  m \left( \ddot{r}_d - \alpha \, \dot{r}_e + g - \mathrm{K} \, \mathfrak{s} \right),
    \end{equation}
    where $\mathrm{K} \in \mathcal{S}^3_+$ is a symmetric positive definite matrix. 
\end{theorem}

\begin{proof}
    For convenience, let $u$ be equal to the right hand side of \cref{eq:Rd}.
    Differentiating \cref{eq:s_pos} and substituting the translational dynamics from \cref{eq:dyn}, we get
    \begin{equation*}
        \begin{aligned}
            \dot{\mathfrak{s}} & = \ddot{r} - \ddot{r}_d + \alpha \, \dot{r}_e \\
            & = \frac{1}{m} R \, T \, \hat{e}_3 - g - \ddot{r}_d + \alpha \, \dot{r}_e.
        \end{aligned}
    \end{equation*}
    Since $R_d \, T \, \hat{e}_3 = m u$, then $T \, \hat{e}_3 = m R_d^\top u$, so the above expression becomes
    \begin{equation*}
        \begin{aligned}
            \dot{\mathfrak{s}} & = R \, R_d^\top \, u - g - \ddot{r}_d + \alpha \, \dot{r}_e.
        \end{aligned}
    \end{equation*}
    Substituting for $u$, 
    \begin{equation*}
        \dot{\mathfrak{s}} = - R \, R_d^\top \mathrm{K} \, \mathfrak{s} + \left(R \, R_d^\top - I\right) \left( \ddot{r}_d - \alpha \, \dot{r}_e + g \right).
    \end{equation*}
    Recalling that $q_e = q_d^* \otimes q \implies R_e = R_d^\top R$, we then have 
    \begin{equation}
        \label{eq:s_dot}
        \dot{\mathfrak{s}} = - R_d \, R_e \, R_d^\top \mathrm{K} \, \mathfrak{s} + \left(R_d \, R_e \, R_d^\top - I\right) \left( \ddot{r}_d - \alpha \, \dot{r}_e + g \right).
    \end{equation}
    Let $V = \|\mathfrak{s}\|^2$ be a candidate Lyapunov function. 
    Then 
    \begin{equation*}
        \begin{aligned}
            \dot{V} = & - 2 \, \mathfrak{s}^\top R_d \, R_e \, R_d^\top \mathrm{K} \, \mathfrak{s} \\
            & +  2 \, \mathfrak{s}^\top\left(R_d \, R_e \, R_d^\top - I\right) \left( \ddot{r}_d - \alpha \, \dot{r}_e + g \right), 
        \end{aligned}
    \end{equation*}
    which can be equivalently written as
    \begin{equation*}
        \begin{aligned}
            \dot{V} =&  - 2 \, \mathfrak{s}^\top \mathrm{K} \, \mathfrak{s} - 2 \, \mathfrak{s}^\top \left( R_d \, R_e \, R_d^\top - I \right) \mathrm{K} \, \mathfrak{s} \\
            &  + 2 \, \mathfrak{s}^\top \left( \left(R_d \, R_e \, R_d^\top - I\right) \left( \ddot{r}_d - \alpha \, \dot{r}_e + g \right) \right).
        \end{aligned}
    \end{equation*}
    Since $\mathrm{K} \in \mathcal{S}^3_+$, there exists $\rho \in \mathbb{R}_{>0}$ such that $ \mathrm{K} \succeq \rho I$, so the above expression can be further simplified to
    \begin{equation*}
        \begin{aligned}
            \dot{V} \leq \, &  - 2 \, \rho \, \| \mathfrak{s}\|^2 - 2 \,\rho \, \mathfrak{s}^\top \left( R_d \, R_e \, R_d^\top - I \right) \mathfrak{s} \\
            &  + 2 \, \mathfrak{s}^\top \left( \left(R_d \, R_e \, R_d^\top - I\right) \left( \ddot{r}_d - \alpha \, \dot{r}_e + g \right) \right).
        \end{aligned}
    \end{equation*}
    Applying Cauchy-Schwarz and using \cref{lemma:frob}, 
    \begin{equation*}
        \begin{aligned}
            \dot{V} \leq &  - 2 \, \rho \, \| \mathfrak{s}\|^2 + 2 \, \rho \, \| R_d \, R_e \, R_d^\top - I \|_F \, \|\mathfrak{s} \|^2 \\
            &  + 2 \, \|R_d \, R_e \, R_d^\top - I \|_F \,  \| \ddot{r}_d - \alpha \, \dot{r}_e + g \| \, \|\mathfrak{s}\| \\
            \end{aligned}
    \end{equation*}
    Using \cref{lemma:rot_frob} and \cref{remark:rot_frob}, $\dot{V}$ becomes
    \begin{equation*}
        \begin{aligned}
            \dot{V} \leq & - 2 \, \rho \left( 1 - 2 \sqrt{2} \|\vec{q}_e(t)\| \right) \|\mathfrak{s}\|^2 \\
                & + 4 \sqrt{2} \|\vec{q}_e(t)\| \|\ddot{r}_d - \alpha \, \dot{r}_e + g \| \, \|\mathfrak{s}\|.
        \end{aligned}
    \end{equation*}
    The above expression can be further simplified by noting that for $\|\mathfrak{s}\| \neq 0$ then $\dot{V} = 2 \, \|\mathfrak{s}\| \frac{d}{dt}\|\mathfrak{s}\|$, so
    \begin{equation*}
        \begin{aligned}
            \frac{d}{dt}\|\mathfrak{s}\| \leq & - \rho \left( 1 - 2 \sqrt{2} \|\vec{q}_e(t)\| \right) \|\mathfrak{s}\| \\
                & + 2 \sqrt{2} \|\vec{q}_e(t)\| \|\ddot{r}_d - \alpha \, \dot{r}_e + g \|.
        \end{aligned}
    \end{equation*}
    When $\|\mathfrak{s}\| = 0$, one can show that $D^+\|s\|$ is equal to the right hand side of the above differential inequality evaluated at $\|\mathfrak{s}\| =0$ as \cref{eq:s_dot} is continuously differentiable in $\mathfrak{s}$ and piecewise continuous in $t$. 
    If we let $\Delta_1(t) = 2 \sqrt{2} \| \vec{q}_e(t)\|$ and $\Delta_2(t) = \Delta_1(t) \| \ddot{r}_d - \alpha \, \dot{r}_e + g \|$, then we get
    \begin{equation}
        \label{eq:s_dot_norm}
        D^+ \| \mathfrak{s} \| \leq -\rho \, (1 - \Delta_1(t)) \|\mathfrak{s} \| + \Delta_2(t),
    \end{equation}    
    which is a perturbed linear differential inequality.    
    With the assumption that the closed-loop attitude dynamics are exponentially convergent, then $ \Delta_1(t) \rightarrow 0$ exponentially.
    Moreover, if $\ddot{r}_d$ and $\dot{r}_e$ are bounded, then $\Delta_2(t) \rightarrow 0$ exponentially. 
    Noting that $\|\mathfrak{s}\|$ and $D^+\|\mathfrak{s}\|$ satisfy the conditions of \cref{lemma:pert_rate,lemma:pert_input}, i.e., $D^+\|\mathfrak{s}\|$ is only discontinuous at 0, we can conclude $\|\mathfrak{s}(t)\| \rightarrow 0$ exponentially. 
\end{proof}

\begin{remark}
    \label{remark:u}
    The right-hand side of \cref{eq:Rd} can be viewed as the preferred control input for the translational dynamics of \cref{eq:dyn}.
    Stating this more explicitly, we have
    \begin{equation}
        \label{eq:pos_control}
        u = \ddot{r}_d - \alpha \, \dot{r}_e + g - \mathrm{K} \, \mathfrak{s}.
    \end{equation}
    If $R = R_d$, then the closed-loop position dynamics are
    \begin{equation*}
        \begin{aligned}
            \ddot{r} & = \frac{1}{m} \, R \, T \, \hat{e}_3 - g \\
            & = \ddot{r}_d - \alpha \, \dot{r}_e + g - \mathrm{K} \, \mathfrak{s} - g,
        \end{aligned}
    \end{equation*}
    which can be shown to be an exponentially stable second-order system for $r_e$, i.e., $ \ddot{r}_e + a_1 \dot{r}_e + a_2 r_e = 0$, given the appropriate selection of $\alpha,\, \lambda$, and $\mathrm{K}$.
\end{remark}

Building on the previous remark, \cref{thm:pos_control} can be strengthened by showing that $r_e$ must converge exponentially to the origin based on the choice of the position sliding variable $\mathfrak{s}$, as shown in the next corollary.

\begin{corollary}
    The closed-loop position dynamics are exponentially convergent with the attitude controller \cref{eq:att_control} and position controller \cref{eq:pos_control}. 
\end{corollary}

\begin{proof}
    Consider the Lyapunov candidate $V = \|r_e\|^2$.
    Differentiating and noting $\dot{r}_e = - \alpha \, r_e + \mathfrak{s}$, then $ \dot{V} = -2 \, \alpha \|r_e\|^2 + 2 \, r_e^\top \mathfrak{s}$, which can be further simplified to obtain $\frac{d}{dt}\|r_e\| \leq - \alpha \, \|r_e\| + \|\mathfrak{s}\|$.
    Following \cref{lemma:pert_input}, since $\lim_{t \rightarrow \infty}\|\mathfrak{s}(t)\| = 0$ exponentially, then so too must $\|r_e\|$.
    Therefore, the closed-loop position dynamics are exponentially convergent using the attitude controller \cref{eq:att_control} and position controller \cref{eq:pos_control}.
\end{proof}

The assumptions for \cref{thm:pos_control} are that i) the closed-loop attitude dynamics are exponentially convergent and ii) there exists a desired rotation matrix $R_d$ and thrust $T$ such that \cref{eq:Rd} holds.
The first condition is met using the attitude controller \cref{eq:att_control}.
The second condition also holds since both $R_d$ and $T$ can be determined from \cref{eq:att_control}.
Specifically, since rotation matrices preserve length, then $T = m \|\ddot{r}_d - \alpha \, \dot{r}_e +g - \mathrm{K} \, \mathfrak{s} \|$.
Determining $R_d$ can be done via geometry.
Using \cref{eq:pos_control}, the problem of determining the desired orientation requires finding $R_d$ such $R_d \, T \, \hat{e}_3 = m u$.
In other words, one must find the rotation matrix $R_d$ that aligns the vectors $T \, \hat{e}_3$ and $m u$.
Geometrically, this can be achieved using the axis-angle rotation representation and Rodrigues' rotation formula.
After normalizing $T \, \hat{e}_3$ and $m u$, the angle of rotation is computed via the dot product, i.e., $\theta = \arccos \left(\hat{e}_3^\top u / \|u\|\right)$, and the axis of rotation via the cross product, i.e., $\hat{n} = \hat{e}_3 \times \hat{u} / \|\hat{e}_3 \times \hat{u}\|$.
Note that the expression for $\hat{n}$ must be carefully computed with $\hat{e}_3$ and $\hat{u}$ being anti-parallel because there are infinitely many rotations that align two anti-parallel vectors. 
When $\hat{e}_3$ and $\hat{u}$ are parallel, then we can set $\hat{n} = 0$.

It is important to note that the desired angular velocity and acceleration vectors $\omega_d$ and $\dot{\omega}_d$ needed by the attitude controller can be determined by differentiating $\cref{eq:Rd}$ after normalization.
Specifically, with $R_d \, \hat{e}_3 = \hat{u}$, then $R_d [\omega_d]_{\times} \hat{e}_3 = \dot{\hat{u}}$, which can simplified via expansion to get
\begin{equation}
    \label{eq:omega_d}
        [\omega_{d,y},\, -\omega_{d,x},\, 0]^\top = R_d^\top \dot{\hat{u}}.
\end{equation}
Differentiating again gives
\begin{equation}
    \label{eq:domega_d}
        [\dot{\omega}_{d,y},\, -\dot{\omega}_{d,x},\, 0]^\top = R_d^\top \ddot{\hat{u}} - R_d^\top [\omega_d]_\times \dot{\hat{u}}.
\end{equation}
Hence, given $u$ from \cref{eq:pos_control}, all the desired quantities needed by the attitude loop, e.g., $q_d$, $\omega_d$, $\dot{\omega}_d$, can be computed.

\subsection{Theoretical Guarantees \& Practical Pitfalls}
In this subsection, we examine the theoretical performance guarantees of exponentially convergent cascaded geometric control and highlight some of its practical limitations.

\textbf{Position Trajectory Tracking.} It is instructive to reiterate the trajectory tracking guarantees of the cascaded geometric control with \cref{eq:att_control,eq:pos_control}: given \emph{any feasible and sufficiently smooth} desired position trajectory $r_d(t)$, it is ensured that $ r(t) \rightarrow r_d(t)$ exponentially as $t \rightarrow \infty$, assuming that the model is known, there are no exogenous disturbances acting on the system, and feedforward is provided to the controller. 
The last condition is especially important because a lack of feedforward violates the conditions of \cref{thm:att_control,thm:pos_control}, and will lead to poor tracking performance.

\textbf{Effects of Attitude Tracking Error.} The closed-loop position sliding variable dynamics in \cref{eq:s_dot_norm} reveals that the attitude tracking error manifests in the outer-loop in two ways: a decrease in the convergence rate of $\|\mathfrak{s}\|$ and as an exogenous disturbance that drives $\|\mathfrak{s}\|$ away from zero. 
In other words, attitude tracking error will cause slower convergence to the desired position trajectory and, more importantly, act as a disturbance to the outer-loop, increasing position tracking error. 
These effects become more problematic if a non-exponentially convergent attitude controller is used in place of \cref{eq:att_control} because most stable attitude controllers can only guarantee bounded error for orientation trajectory tracking. 
Note that these effects also occur when the attitude loop is counteracting an exogenous torque disturbance, even with an exponentially convergent attitude tracking controller.

\textbf{Inherent Robustness.} Because the closed-loop translational dynamics are exponentially convergent, the closed-loop system exhibits desirable robustness to parametric and non-parametric uncertainties. 
For example, consider the scenario where an unknown but bounded exogenous disturbance $d \in \mathcal{D} \subset \mathbb{R}^3$ is added to the translational dynamics of \cref{eq:dyn}.
\Cref{eq:s_dot_norm} becomes $D^+\|\mathfrak{s}\| \leq -\rho(1-\Delta_1(t)) \| \mathfrak{s}\| + \Delta_2(t) + \|d\|$. 
Since $\|d\|$ is bounded by assumption, it can be shown that $\|s\|$ (and hence $\|r_e\|$) will exponentially converge to a bounded region around zero, which is considered a form of robustness for nonlinear systems. 
Using terminology from the model predictive control literature \cite{langson2004robust,mayne2011tube,lopez2019dynamic}, exponential convergence implies the existence of a control-invariant tube centered around the nominal/desired trajectory that can be used to robustly tighten constraints. 
This guarantees that the constraints are always satisfied despite the presence of uncertainties.
The analysis presented here shows that the trajectory tracking guarantees inherent to the presented cascaded geometric controller can be used to tighten safety constraints in a provably-safe manner, guaranteeing safe closed-loop trajectory design and execution. 

\textbf{Matched Model Uncertainties.} Model uncertainties and exogenous disturbances can be classified as either matched or unmatched \cite{krstic1995nonlinear,lopez2023dynamic}.
Formally, matched uncertainties lie within the span of the control input matrix, while unmatched uncertainties lie outside it.
Stated another way, a matched uncertainty can always be canceled directly through control, but an unmatched uncertainty cannot. 
Again, let a bounded model uncertainty/disturbance $d \in \mathcal{D} \subset \mathbb{R}^3$ be added to the translational dynamics of \cref{eq:dyn}.
If we let $R_d \, T \, \hat{e}_3 = m u$, the translational dynamics are 
\begin{equation*}
    \ddot{r} = R_d R_e R_d^\top u - g + d,
\end{equation*}
with $R_d R_e R_d^\top$ being the perturbed control input matrix (due to the presence of $R_e$).
In the nominal case, i.e., when $R = R_d \implies R_e = I$, the control input matrix is the identity, and the uncertainty $d$ is clearly within its span.
If $R_e \neq I$, the perturbed control input matrix still spans $\mathbb{R}^3$ since $\det(R_d R_e R_d^\top) = \det(R_d) \det(R_e) \det(R_d) = 1$ so $d$ is still matched. 
Hence, the uncertainty $d$ is \emph{always matched} with cascaded geometric control. 
As a result, all existing robust or adaptive control strategies developed for matched uncertainties can be simply appended to \cref{eq:pos_control} as if the system were a multi-dimensional double integrator. 
For example, if the model uncertainty came in the form of aerodynamic drag, i.e., $d = -c_d \|\dot{r}\|\dot{r}$ with unknown coefficient $c_d$, then a robust control law that guarantees a bounded tracking error is \cite{krstic1995nonlinear} 
\begin{equation}
    \label{eq:u_drag}
    u = \ddot{r}_d - \alpha \, \dot{r}_e + g - \mathrm{K} \, \mathfrak{s} + \hat{c}_d \|\dot{r}\| \dot{r} - \frac{1}{4} \|\dot{r}\|^4 \mathfrak{s} 
\end{equation}
where $\hat{c}_d$ is the nominal value of the unknown coefficient $c_d$.
Hence, any disturbance that acts on the rotational or translational dynamics of \cref{eq:dyn} is matched using cascaded geometric control. 
Although this is ideal from a robustness perspective, it does cause practical issues, as discussed next.

\textbf{Practical Limitations.} Three practical limitations are worth discussing. 
The first is reliance on vehicle/model parameters to achieve perfect trajectory tracking, especially in the attitude controller, which requires knowing the inertial tensor $J$. 
Despite this, the presented cascaded geometric controller will still exhibit good performance in practice with an imperfect model because of its inherent robustness (see previous discussion).
The second is a more critical issue often overlooked with this control strategy: \emph{the need for acceleration and jerk feedback}. 
Feedback on acceleration, i.e., feedback on $\ddot{r}_e = \ddot{r} - \ddot{r}_d$, is needed to compute $\omega_d$ for the attitude controller, as given by \cref{eq:omega_d} for $u$ defined by \cref{eq:pos_control}.
Similarly, feedback on jerk, i.e., feedback on ${r}^{(3)}_e = {r}^{(3)} - {r}^{(3)}_d$, is required to compute $\dot{\omega}_d$ for the attitude controller via \cref{eq:domega_d}.
In practice, computing these feedback terms is often done approximately via numerical differentiation of velocity tracking error because acceleration and jerk are noisy or not directly measured, which introduces errors that propagate through the closed-loop system, degrading tracking performance.
The third is related to the previous limitation, namely the need to differentiate feedforward or model inversion terms, which again requires access to higher-order derivative state information.
For example, the robust feedback controller \cref{eq:u_drag} that counteracts drag will produce terms related to $\ddot{r}$ and $r^{(3)}$ when differentiated, so these terms must be measured or estimated to implement the controller.
For more complex feedforward models, such as those learned from data, differentiation can be problematic because any model noise/error is amplified and propagated through the closed-loop system.
The importance of feedforward in nonlinear control is well-documented \cite[pg.~199]{slotine1991applied}, and has recently been a topic of discussion in relation to the performance of classical controllers compared to learned controllers (see \cite{kunapuli2025leveling} and the discussion therein).
However, the need to differentiate feedforward terms, along with the other limitation discussed above, all due to the nested architecture of cascaded geometric control, raises a broader question of whether new control strategies are needed that are better suited for highly-dynamic vehicles where higher-order effects, e.g., unsteady aerodynamics, come into play.

\section{Concluding Remarks}
A new stability proof for cascaded geometric control was presented for a class of aerial vehicles tasked with tracking a time-varying position trajectory.
The novelty of the approach stems from using a recently proposed quaternion-based attitude sliding controller that, when nested with a linear position controller, yielded an exponentially convergent closed-loop system.  
The use of sliding variables in our analysis not only streamlined that stability proof, but revealed how attitude tracking error affects position tracking. 
Several new performance aspects of the controller were also discussed, such as its inherent robustness and how uncertainties enter the systems, in addition to some of the practical pitfalls of the approach, e.g., feedback on higher-order derivatives of position.
Two broader implications of the results will be explored in future work. 
The first involves developing a provably safe trajectory planning approach like that in \cite{lopez2019dynamic} by combining the dynamic tracking error bounds derived in this work with recent advances in real-time sampling-based trajectory planning \cite{levy2024stitcher,levy2025stitcher}. 
The second focuses on developing a new control strategy that inherits the same performance guarantees as cascaded geometric control, but addresses the practical limitations discussed earlier.

\textbf{Acknowledgments.} The authors thank Jean-Jacque Slotine for his insightful comments early on in this work. 

\section{Appendix}
\begin{lemma}
\label{lemma:ode}
    Let $x$ be restricted to the interval $x \in [0,\,1)$.
    If $x$ satisfies the differential equation
    \begin{equation*}
        \dot{x} = - \sigma \, x \, \sqrt{1-x},
    \end{equation*}
    for $\sigma \in \mathbb{R}_{>0}$ then $x(t) \rightarrow 0$ exponentially as $t \rightarrow \infty$.   
\end{lemma}
\begin{proof}
    The above differential equation is separable.
    Using the substitution $ y = \sqrt{1-x}$, one gets
    \begin{equation*}
        \int_{y_0}^{y} \frac{dz}{z^2 - 1} = - \sigma \, t,
    \end{equation*}
    which has the solution
    \begin{equation*}
        y(t) = \frac{1-c\,e^{-\sigma t}}{1 + c \, e^{-\sigma t}} ~~ \implies ~~ x(t) = \frac{4\, c \, e^{-\sigma t}}{(1+c\,e^{-\sigma t})^2} \leq  {4\, c \, e^{-\sigma t}}
    \end{equation*}
    for $c = (1-\sqrt{1-x_0}) / (1+\sqrt{1-x_0})$. 
    It is clear that $x$ tends to zero exponentially at rate $\sigma$, as desired.
\end{proof}

\begin{lemma}
    \label{lemma:frob}
    Let $A \in \mathbb{R}^{m\times n}$ and $x \in \mathbb{R}^n$. 
    Then $\|A x\| \leq \|A\|_F \|x\|$ where $\|\cdot\|_F$ is the Frobenius norm of $A$.
\end{lemma}
\begin{proof}
    Let the $i$-th row vector of $A$ be $a_i^\top \in \mathbb{R}^n$.
    The product $Ax$ can then be expressed as $Ax = [a_1^\top x, \dots,  a_m^\top x]^\top$.
    Taking the norm gives $\|A x\|^2 = (a_1^\top x)^2 + \dots + (a_m^\top x)^2$, and by Cauchy-Schwarz, $\|a_i^\top x \|^2 \leq \|a_i\|^2 \|x\|^2$, so $\| A x \|^2 \leq \|a_1\|^2\|x\|^2 + \dots + \|a_m\|^2 \|x\|^2 = \sum_i^m \|a_i\|^2 \|x\|^2$.
    Since $a_i$ is the $i$-th row vector of $A$, then $\|a_i\|^2 = \sum_j^n a_{ij}^2$.
    Hence, $\|A x\|^2 \leq \sum_i^m \sum_j^n a_{ij}^2 \|x\|^2 = \|A\|_F^2 \|x\|^2$.
    Taking the square root gives the desired result.
\end{proof}

\begin{lemma}
    \label{lemma:pert_rate}
    Suppose $v$ is continuous, non-negative uniformly, and satisfies the differential inequality
    \begin{equation}
        \label{eq:ode1}
        D^+{v} \leq (- \lambda + \Delta(t) ) \, v
    \end{equation}
    where $\lambda \in \mathbb{R}_{>0}$ and $0 \leq \Delta(t) < \infty$ uniformly.
    If $\Delta(t) \rightarrow 0$ as $t \rightarrow \infty$ and $D^+v$ is discontinuous on a set of measure zero, then $v(t) \rightarrow 0$ exponentially as $t \rightarrow \infty$.
\end{lemma}
\begin{proof}
    Since $\Delta(t) \rightarrow 0$, for $0 < \epsilon < \lambda$ there exists a $\delta > 0$ such that for $t > \delta$ then $\Delta(t) < \epsilon < \lambda$.
    Hence, there exists $\gamma = \lambda - \epsilon \in \mathbb{R}_{>0}$, so $D^+{v} \leq - \gamma v$.    
    Given the conditions on $v$ and $D^+v$, the solution to \cref{eq:ode1} is $v(t) \leq v(0) e^{-\gamma t}$ \cite{hagood2006recovering}. 
    Therefore, $v(t) \rightarrow 0$ exponentially, as desired.
\end{proof}

\begin{remark}
    Note that if $\Delta(t) \rightarrow 0$ exponentially as $t \rightarrow \infty$, i.e., $\Delta(t) = C e^{-p t}$, then one can show $v(t) \leq v(0) \, e^{-\lambda t + C (1 - e^{-pt})/p} \leq R \, v(0) \, e^{-\lambda t}$
    where $R = e^{C/p}$ is referred to as overshoot.
\end{remark}

\begin{lemma}
    \label{lemma:pert_input}
    Suppose $v$ is continuous, non-negative uniformly, and satisfies the differential inequality
    \begin{equation}
        \label{eq:ode2}
        D^+{v} \leq - \lambda  \, v + \Delta(t)
    \end{equation}
    where $\lambda \in \mathbb{R}_{>0}$, $0 \leq \Delta(t) < \infty$ uniformly.
    If $\Delta(t) \rightarrow 0$ as $t \rightarrow \infty$ and $D^+v$ is discontinuous on a set of measure zero, then $v(t) \rightarrow 0$ as $t \rightarrow \infty$.
\end{lemma}

\begin{proof}
    Suppose $D = \sup_{t \geq \tau_1} \Delta(t) < \infty$ for $\tau_1 \geq 0$.
    Given the conditions on $v$ and $D^+v$, the solution to \cref{eq:ode2} is \cite{hagood2006recovering}
    \begin{equation*}
        v(t) \leq v(0) \, e^{-\lambda (t-\tau_1)} + \frac{D}{\lambda} \left(1 - e^{\lambda (t-\tau_1)}\right).
    \end{equation*}
    Since $\Delta(t) \rightarrow 0$ then $D \rightarrow 0$ as $\tau_1 \rightarrow \infty$.
    Therefore, $v(t) \rightarrow 0$ as $t \rightarrow 0$, as desired.
\end{proof}

\begin{corollary}
    If $\Delta(t) \rightarrow 0$ exponentially as $t \rightarrow \infty$, then $v(t) \rightarrow 0$ exponentially as $t \rightarrow \infty$.
\end{corollary}

\begin{proof}
    Let $\Delta(t) = C e^{-pt}$ where $p, C \in \mathbb{R}_{>0}$ and $p \neq \lambda$.
    Given the conditions on $v$ and $D^+v$, the solution to the differential inequality in \cref{lemma:pert_input} is \cite{hagood2006recovering} 
    \begin{equation*}
        v(t) \leq v(0) e^{-\lambda t} + \frac{C}{\lambda - p} \left( e^{-p t} - e^{-\lambda t} \right),
    \end{equation*}
    so $v(t) \rightarrow 0$ exponentially at rate $\min\{\lambda,\,p\}$ as $t \rightarrow 0$.
\end{proof}

\bibliographystyle{ieeetr}
\bibliography{ref}

\end{document}